\documentclass[final,11pt]{amsart}
\usepackage{amsmath,amssymb,amsfonts} 
\usepackage{epsfig} \usepackage{latexsym,nicefrac,bbm}
\usepackage{xspace}
\usepackage{color,fancybox,graphicx,subfigure,fullpage}
\usepackage[top=1in, bottom=1in, left=1in, right=1in]{geometry}
\usepackage{tabularx} \usepackage{hyperref} 
\usepackage{pdfsync}
\usepackage[boxruled,norelsize]{algorithm2e}
\usepackage{multicol}
\newcommand{\parag}[1]{ {\bf \noindent #1}}
\newcommand{\defeq}{\stackrel{\textup{def}}{=}}
\newcommand{\nfrac}{\nicefrac}
\newcommand{\eps}{\epsilon}

\newcommand{\barr}{\mathcal{B}}

\newcommand{\argmin}{\operatornamewithlimits{argmin}}

\renewcommand{\epsilon}{\varepsilon}

\def\showauthornotes{1} 
 
\def\showdraftbox{1}

\newtheorem{theorem}{Theorem}[section]
\newtheorem{conjecture}[theorem]{Conjecture}

\newtheorem{lemma}[theorem]{Lemma}
\newtheorem{remark}[theorem]{Remark}

\newtheorem{corollary}[theorem]{Corollary}
\newtheorem{claim}[theorem]{Claim}
\newtheorem{fact}[theorem]{Fact}

\def\FullBox{\hbox{\vrule width 6pt height 6pt depth 0pt}}

\def\qed{\ifmmode\qquad\FullBox\else{\unskip\nobreak\hfil
\penalty50\hskip1em\null\nobreak\hfil\FullBox
\parfillskip=0pt\finalhyphendemerits=0\endgraf}\fi}

\def\qedsketch{\ifmmode\Box\else{\unskip\nobreak\hfil
\penalty50\hskip1em\null\nobreak\hfil$\Box$
\parfillskip=0pt\finalhyphendemerits=0\endgraf}\fi}

\newenvironment{proofof}[1]{\begin{trivlist} \item {\bf Proof
#1:~~}}
  {\qed\end{trivlist}}

\newcommand\N{\mathbb N}
\newcommand\R{\mathbb R}

\newcommand{\marginlabel}[1]%
{\mbox{}\marginpar{\it{\raggedleft\hspace{0pt}#1}}}

\definecolor{Mygray}{gray}{0.8}

 \ifcsname ifcommentflag\endcsname\else
  \expandafter\let\csname ifcommentflag\expandafter\endcsname
                  \csname iffalse\endcsname
\fi

\ifnum\showauthornotes=1
\newcommand{\todo}[1]{\colorbox{Mygray}{\color{red}#1}}
\else
\newcommand{\todo}[1]{}
\fi

\ifnum\showauthornotes=1
\newcommand{\Authornote}[2]{{\sf\small\color{red}{[#1: #2]}}}
\newcommand{\Authoredit}[2]{{\sf\small\color{red}{[#1]}\color{blue}{#2}}}
\newcommand{\Authorfnote}[2]{\footnote{\color{red}{#1: #2}}}
\newcommand{\Authorfixme}[1]{\Authornote{#1}{\textbf{??}}}
\newcommand{\Authormarginmark}[1]{\marginpar{\textcolor{red}{\fbox{
#1:!}}}}
\else
\newcommand{\Authornote}[2]{}
\newcommand{\Authoredit}[2]{}
\newcommand{\Authorcomment}[2]{}
\newcommand{\Authorfnote}[2]{}
\newcommand{\Authorfixme}[1]{}
\newcommand{\Authormarginmark}[1]{}
\fi

\newcommand{\inparen}[1]{\left(#1\right)}             
\newcommand{\inbraces}[1]{\left\{#1\right\}}           

\def\abs#1{\left| #1 \right|}
\newcommand{\norm}[1]{\ensuremath{\left\lVert #1 \right\rVert}}

\newcommand{\diag}[1]{{\sf Diag}\left({#1}\right)}

 {
	\begin{enumerate}}{\end{enumerate}}

\newcounter{lecnum}

\newlength{\tpush}
\ifnum\showdraftbox=1

\else

\fi

\title{\bf   IRLS and  Slime Mold: \\ { Equivalence and  Convergence}}

\author{Damian Straszak}
\thanks{Damian Straszak, \'{E}cole Polytechnique F\'{e}d\'{e}rale de Lausanne (EPFL)} 

\author{Nisheeth K. Vishnoi}
\thanks{Nisheeth K. Vishnoi, \'{E}cole Polytechnique F\'{e}d\'{e}rale de Lausanne (EPFL)}

\begin{document}
\maketitle

\begin{abstract}
In this paper we present a connection between two dynamical systems arising in entirely different contexts: one in signal processing  and the other in biology. The first is the famous Iteratively Reweighted Least Squares (IRLS) algorithm used in compressed sensing and sparse recovery while the second is the dynamics of a slime mold ({\em Physarum polycephalum}). Both of these dynamics are geared towards  finding a minimum $\ell_1$-norm solution in an affine subspace.  Despite its simplicity the convergence of the IRLS method has been shown  only for a certain {\em regularization} of it and remains an important open problem \cite{Beck15,Daubechies10}. 
Our first result shows that the two dynamics are projections of the {\em same} dynamical system in higher dimensions. As a consequence, and building  on the recent work on Physarum dynamics, we are able to prove convergence and obtain complexity bounds for a {\em damped} version of the IRLS algorithm.

\end{abstract}
\maketitle

\tableofcontents

\thispagestyle{empty}

\newpage

\setcounter{page}{1}

\section{Introduction}

\paragraph{\bf Sparse recovery and basis pursuit} 
A classical task in signal processing is to recover a sparse signal from a small number of linear measurements. Mathematically, this  can be formulated as the problem of finding a solution  to a linear system $Ax=b$ where  $A\in \R^{m \times n}, \; b \in \mathbb{R}^m$ are given and $A$ has far fewer rows than columns (i.e. $m\ll n$). Among all the solutions, one would like to recover one with the fewest  non-zero entries. This problem, known as sparse recovery, is NP-hard and  we cannot hope to find an efficient algorithm in general. However, it has been observed experimentally that, when dealing with real-world data, a solution to the following $\ell_1$-minimization problem (also known as {\em basis pursuit}):
\begin{equation}\label{l1_min_intro}
\min \; \norm{x}_1 \qquad \mathrm{s.t. } \; \; Ax=b
\end{equation}
is typically quite sparse, if not of optimal sparsity. The history of theoretical investigations on how to explain the above phenomenon is particularly rich. It was first shown in~\cite{DH01, DE02} that the $\ell_1$-norm objective is in fact equivalent to sparsity for a specific family of matrices  and, later, the same was argued for a class of random matrices~\cite{CRT06}. Finally, the notion of Restricted Isometry Property (RIP) was formulated in~\cite{CT05} and shown to guarantee sparse recovery via~\eqref{l1_min_intro}. 
Consequently, optimization problems of the form~\eqref{l1_min_intro} became important building blocks for applications in signal processing and statistics. Thus, fast algorithms for solving such problems are desired. Note that~\eqref{l1_min_intro} can be cast as a linear program of size linear in $n$ and $m$ and, hence, any linear programming algorithm can be used to solve it. However, because of the special structure of the problem, many algorithms were developed which outperform standard LP solvers in terms of efficiency on real world instances. To make an algorithm applicable in practice another property is highly desirable: simplicity. This is not only for the ease of implementation, but also due to the fact that simple solutions are typically more robust and extendable to slightly different settings, such as noise tolerance. 

\medskip
\paragraph{\bf Iteratively Reweighted Least Squares} One of the simplest algorithms for solving problem~\eqref{l1_min_intro} is the  Iteratively Reweighted Least Squares algorithm (IRLS). IRLS is a very general scheme for solving optimization problems: it produces a sequence of points $y^{(0)}, y^{(1)}, \ldots$ with every $y^{(k+1)}$ obtained as a result of solving a weighted $\ell_2$-minimization problem, where the weights are appropriately chosen based on the previous point $y^{(k)}$. Let us now describe one extremely popular scheme of this kind, which is of main focus in this paper. We pick any starting point $y^{(0)} \in \R^n$. Then, $y^{(k+1)}$ is obtained from $y^{(k)}$ as the solution to the following optimization problem: 
\begin{equation}\label{intro_irls}
y^{(k+1)} \defeq \argmin_{x\in \R^n} \sum_{i=1}^n \frac{x_i^2}{\abs{y_i^{(k)}}} \qquad \mathrm{s.t. }  \; \; Ax=b.
\end{equation}
For this optimization problem to make sense, we need to assume  that $y_i^{(k)}\neq 0$ for every $i=1,2,\ldots, n$; however, the above can be made formal without this assumption. Importantly, the weighted $\ell_2$-minimization in~\eqref{intro_irls} can be solved via a formula which involves solving a linear system. The resulting  algorithm does not require any preprocessing of the data or any special rules for choosing a starting point. These properties make the algorithm particularly attractive for practical use and, indeed, the IRLS algorithm is quite popular; see for instance ~\cite{CY08,Green84}.
    However, from a theoretical viewpoint, the algorithm is still far from being understood. No global convergence analysis is known.
    One can construct examples to show that  there may be starting points for which the IRLS algorithm does not converge, see Appendix~\ref{app:nonconvergence}.
      The only known rigorous positive results~\cite{Osborne85}  concern the case when the algorithm is initialized very close to the optimal solution. It is an important open problem to establish global convergence of the IRLS algorithm. 

\medskip
\paragraph{\bf The dynamics of a slime mold.} In a seemingly unrelated story,   in 2000  a striking experiment demonstrated that a slime mold ({\it Physarum polycephalum}) can solve the shortest path problem in a maze \cite{NTY00}. The need to explain how, resulted in a mathematical model~\cite{TKN07} which was a dynamical system;  we (loosely) refer to this dynamical system as Physarum dynamics. 
Subsequently, this model was successfully analyzed mathematically and generalized to many different graph problems (\cite{MO07,IJNT11,BMV12,BBDKM13,SV15}). 
In this work we propose an extension of the Physarum dynamics for solving the basis pursuit problem.   Given $A,b$ as before, we let $w^{(0)}\in \R_{>0}^n$ to be any point with positive coordinates and pick any step size $h\in (0,1)$. The discrete Physarum dynamics iterates according to the following formula:
\begin{equation}\label{intro_physarum}
w^{(k+1)} \defeq (1-h) w^{(k)} + h \abs{q^{(k)}}.
\end{equation}
In the above, $q^{(k)}$ is the vector that minimizes $\sum_{i=1}^n \frac{x_i^2}{w_i^{(k)}}$ over all $x\in \R^n$ such that $Ax=b$. The absolute value of $q^{(k)}$ should be understood entry-wise. The above is a generalization of the Physarum dynamics for the shortest $s-t$ path problem in an undirected graph~\cite{TKN07}, for which
it was shown by \cite{BMV12} that $w^{(k)}$ converges to the characteristic vector of the shortest $s-t$ path in $G$. 
Interestingly, since $w^{(k)}$ remains a positive vector at every step $k$, the vector $w^{(k)}$ may not converge to the optimal solution. In Section \ref{sec:comparison} we explain how to define an auxiliary sequence $y^{(k)}$ which converges to an optimal solution.

\paragraph{\bf IRLS vs. Physarum.}
Both algorithms, IRLS and Physarum can be seen as discrete dynamical systems, with updates based on a certain weighted $\ell_2$-minimization, however no formal relation between them is apparent.  Our first result connects these two algorithms. Both of these algorithms are naturally viewed  as discrete dynamical systems over a $2n$-dimensional domain $\Gamma \defeq \{(y,w): y\in \R^n, w\in \R_{> 0}^n\}$ with the vector field $F:\Gamma \to \R^n\times \R^n$ defined as:
\begin{align}\label{vector_field}
\begin{split}
F(y,w) &\defeq (q-y,|q|-w)\\
 q& \defeq \argmin_{x\in \R^n} \sum_{i=1}^n \frac{x_i^2}{w_i} \qquad \mathrm{s.t. } \; \;  Ax=b.
\end{split}
\end{align}
More precisely we prove the following theorem in Section~\ref{sec:comparison}.

\begin{theorem}[Informal]\label{thm:main1}
Given a starting point $(y^{(0)}, w^{(0)})\in \Gamma$ and $h\in (0,1]$ let us consider the sequence $\inbraces{(y^{(k)}, w^{(k)})}_{k\in \N}$ generated by taking steps in the direction suggested by $F$:
\begin{equation}\label{general_update}
({y}^{(k+1)}, {w}^{(k+1)})=(1-h)({y}^{(k)}, {w}^{(k)})+hF({y}^{(k)}, {w}^{(k)}). 
\end{equation}
When $h=1$ the sequence $\inbraces{{y}^{(k)}}_{k\in \N}$ is identical to that produced by IRLS, while for $h\in (0,1)$, the sequence $\inbraces{{w}^{(k)}}_{k\in \N}$ is equivalent to Physarum dynamics.
\end{theorem}

\noindent
The above tells us additionally that IRLS and Physarum are complimentary in terms of their descriptions, since the variables $y$ appear in the definition of IRLS, while $w$ are implicit (and vice versa for Physarum).

Our second contribution is a global convergence analysis for Physarum dynamics which implies the same for the {\em damped} version of IRLS. We state it informally below; several details are omitted and only the dependence on $\eps$ is emphasized, the quantities depending on the dimension and the input data are denoted by $C_1$ and $C_2$. For a precise formulation we refer to Theorem~\ref{thm:convergence}.

\begin{theorem}[Informal]\label{thm:main2}
Suppose we initialize the Physarum dynamics at an appropriate point $w^{(0)}$. Take an arbitrary $\eps>0$ and choose $h\leq \frac{\eps}{C_1}$. If we generate a sequence $\inbraces{w^{(k)}}_{k\in \N}$ according to the Physarum dynamics~\eqref{intro_physarum}, then after $k=\frac{C_2}{h\eps^2}$ steps one can compute a vector $y^{(k)}\in \R^n$ such that $Ay^{(k)}=b$ and $\norm{y^{(k)}}_1 \leq \norm{w^{(k)}}_1\leq \norm{x^\star}_1 \cdot (1+\eps),$
where $x^\star$ is any optimal solution to~\eqref{l1_min_intro}.
\end{theorem}

\section{Related Work} 
 
\parag{IRLS.} Many different algorithms based on IRLS have been proposed for solving a variety of optimization problems. The book~\cite{Osborne85} presents (among others) the IRLS method for $\ell_1$-minimization and proves a local convergence result (assuming the starting point is sufficiently close to the optimum and no zero-entries appear in the iterates). The paper~\cite{GR97} discusses a number of different IRLS schemes for finding sparse solutions to underdetermined linear systems. It provides convergence results for a family of such methods, but the algorithm studied in our paper is not covered. In \cite{RKD99},  IRLS schemes for minimizing $\inparen{\sum_{i=1}^n \abs{x_i}^p}^{1/p}$ are proposed, the scheme given for $p=1$ matches our setting, however no global convergence results are obtained.

 We now discuss another line of work, for which rigorous convergence results are known. To circumvent mathematical difficulties related to zero-entries appearing in IRLS iterates one can choose a small positive constant $\eta>0$ and define a modified version of the IRLS update:
\begin{equation}\label{irls_update_reg}
x^{(k+1)} =  \argmin \inbraces{\sum_{i=1}^n \frac{x_i^2}{\sqrt{\abs{x_i^{(k)}}^2+\eta^2}}: x\in \R^n, Ax=b}.
\end{equation}
Note that the above minimization problem makes perfect sense even when $x_i^{(k)}=0$ for some $i$. Consequently, it has a unique solution, for every choice of $x^{(k)}$. It was proved in~\cite{Beck15} that the sequence of points produced by scheme~\eqref{irls_update_reg} converges to the optimal solution of:
\begin{align}\label{l1_min_reg}
\begin{split}
\min \; \quad & \sum_{i=1}^n \inparen{x_i^2+\eta^2}^{1/2} \\
\mathrm{s.t. } \quad & Ax=b.\\ 
\end{split}
\end{align}
The number of iterations required to get $\eps$-close to the optimal solution is bounded by $O\inparen{\frac{C}{\eps }}$, where $C$ is a quantity depending on $A,b$ and $\eta$. 

The function $\sum_{i=1}^n \inparen{x_i^2+\eta^2}^{1/2}$ approximates the $\ell_1$ norm in the following sense:
$$\forall x\in \R^n \qquad \norm{x}_1 \leq \sum_{i=1}^n \inparen{x_i^2+\eta^2}^{1/2} \leq \norm{x}_1+n\cdot \eta.$$
\noindent
In the case when the matrix $A$ satisfies a variant of RIP (Restricted Isometry Property), \cite{Daubechies10} showed that a scheme similar to~\eqref{l1_min_reg} (with $\eta_k \to 0$ in place of constant $\eta$) converges to the $\ell_1$-optimizer. The proof relies on non-constructive arguments (compactness is repeatedly used to obtain certain  accumulation points) hence no quantitative bounds on the global convergence rate follow from this analysis.

\parag{Physarum dynamics.}
The discrete Physarum dynamics we propose for the basis pursuit problem can be seen as an analogue of the similarly looking, but technically very different, dynamics for linear programming studied in~\cite{JZ12, SVLP15}. Our second main result  (Theorem \ref{thm:convergence}) builds up,  extends and simplifies a recent result \cite{SV15} of the authors for the case of {\em flows}; when the matrix $A$ corresponds to an incidence matrix of an undirected graph. For more on prior work on Physarum dynamics, the reader is referred to \cite{SV15}.

\newpage

\section{Preliminaries}\label{section:preliminaries}
\medskip
\parag{Notation for sets, vectors and matrices.}
The set $\{1,2,\ldots,n\}$ is denoted by $[n]$. All vectors considered  are column vectors. By $x\in \R^Q$, for some finite set $Q$, we mean a $|Q|$-dimensional real vector indexed by elements of $Q$, similarly for matrices. If $x\in \R^n$ is a vector then $x_S$ for $S\subseteq [n]$ denotes a vector in $\R^S$ which is the restriction of $x$ to indices in $S$.

The basis pursuit problem is to find  a minimum $\ell_1$-norm solution to the linear system $Ax=b$, where $A$ is an $m\times n$ matrix. We assume that $A$ has rank $m$.\footnote{It is enough here to assume $b\in \mathrm{Im}(A)$ only. To simplify notation, we work with the full-rank assumption. One can reduce the general case to full-rank by removing some number of rows from $A$. Both dynamics remain the same.}The $i$-th column of $A$ is denoted by $a_i\in \R^m$.

If $x\in \R^n$ then by $X$ we mean an $n\times n$ real diagonal matrix with $x$ on the diagonal, i.e. $X=\diag{x}$. 
Whenever $x\in \R^n$ is a vector and a scalar operation is applied to it, the result is a vector with this scalar operation applied to every entry. For example $\abs{x}$ denotes a vector $y\in \R^n$ with $y_i = |x_i|$ for every $i\in [n]$. When writing inequalities between vectors, like $x\leq y$ (for $x,y\in \R^n$) we mean that $x_i\leq y_i$ for all $i\in [n]$, also $x>0$ means $x_i>0$ for every $i\in [n]$.

For a symmetric matrix $M\in \R^{d\times d}$ we denote by $M^{+}$ its Moore-Penrose pseudoinverse. It satisfies  $MM^{+}x=M^{+}Mx=x$ for every $x\in \R^d$ from the image of $M$. 

\medskip
\parag{Weighted $\ell_2$-minimization.}
The weighted $\ell_2$-minimization problem is the following: for a given matrix $A\in \R^{m\times n},$ vector $b\in \R^m$ and weights $s\in \R^n_{>0}$ find:
\begin{align*}
\argmin_{x\in \R^n} \sum_{i=1}^n s_i x_i^2 \qquad 
\mathrm{s.t. }\; \;  A x=b.
\end{align*}
One can show that if the linear system $Ax=b$ has a solution, then the above has a unique solution $q\in \R^n$ which can be computed as:
$$q = SA^\top (ASA^\top)^+b.$$

\subsection{The IRLS algorithm}
We now present the Iteratively Reweighted Least Squares (IRLS) algorithm.\footnote{As mentioned before, IRLS is in fact a general algorithm scheme; however, in the remaining part of the paper by IRLS we always mean the specific IRLS for basis pursuit.} For readability,  some technical details are omitted; however, we leave remarks wherever additional care is required. Consider the basis pursuit problem:

\begin{align}\label{l1_min}
\begin{split}
\min \; \quad & \norm{x}_1 \\
\mathrm{s.t. }\; \quad & Ax=b
\end{split}
\end{align}
where $x\in \R^n$, $A\in \R^{m\times n}$ and $b\in \R^m$. The algorithm starts from an arbitrary point $y^{(0)}\in \R^n$, e.g. $y^{(0)}=(1,1,\ldots,1)^\top,$ and performs the following iterations for $k=0,1,2,\ldots:$
\begin{equation}\label{irls_update}
y^{(k+1)}= \argmin \inbraces{\sum_{i=1}^n \frac{x_i^2}{\abs{y_i^{(k)}}}: x\in \R^n, Ax=b}.
\end{equation}
Thus, the new point is a result of $\ell_2$-minimization with weights coming from the previous iteration. 
\begin{remark}
Note that the above is well defined only if $y_i^{(k)}\neq 0$ for every $i\in [n]$. Additional care is required to deal with the case where some $y_i^{(k)}$ are zero. Informally, one can imagine that if $y_i^{(k)}=0$ then the weight on the $i$-th coordinate is $+\infty$; hence, one is forced to choose $x_i=0$. In fact the formal treatment follows this intuition: whenever $y_i^{(k)}=0$, one adds a hard constraint $x_i=0$ and performs the weighted $\ell_2$-norm minimization over the non-zero coordinates.
\end{remark}

\noindent
We remark that the $\ell_2$-minimization problem in the update rule has a closed form solution involving a projection:
$$y^{(k+1)} = Y^{(k)}A^\top(AY^{(k)} A^\top)^{+}b,$$
where $Y^{(k)}\defeq \diag{y^{(k)}}.$

It is easy to show that the $\ell_1$-norm of the subsequent iterates $y^{(1)}, y^{(2)}, \ldots$ is non-increasing; however,  this does not necessarily imply that IRLS converges to the optimal solution. In fact no result on global convergence (to an optimal solution to~\eqref{l1_min}) is known for IRLS. While the convergence is indeed observed in practice, it remains open to prove this. One issue is that there are examples of instances and starting points, where the sequence provably does not converge to the optimal solution; see Appendix~\ref{app:nonconvergence}. 
However, we believe that the following conjecture might hold  regarding the convergence of IRLS. 
\begin{conjecture}
The set of starting points $y^{(0)}\in \R^n$ for which the sequence $\inbraces{y^{(k)}}_{k\in \N}$ generated by IRLS does not converge to an optimal solution to~\eqref{l1_min} is of measure zero.
\end{conjecture}

One of the main obstacles in proving global convergence for IRLS is its ``non-uniform'' behavior, depending on the support of the current point. Unfortunately, the issue of $y^{(k)}$ having zero-entries cannot be avoided. Note that this problem is not only a mathematical inconvenience. In fact, when dealing with instances where one or more entries of $y^{(k)}$ are close to zero, numerical issues are likely to appear. When solving a minimization problem of the kind~\eqref{irls_update}, tiny values of $y_i^{(k)}$ can be unpleasant to deal with and cause errors.

\subsection{Continuous Physarum dynamics for $\ell_1$-minimization}

The Physarum dynamics was originally introduced for an undirected graph $G=(V,E)$ as a continuous time dynamical system over $\R_{>0}^E$ (\cite{TKN07}). This model was proposed to explain the experimentally observed ability of Physarum to solve the shortest path problem. It was then extended to a more general flow problem: the transshipment problem (\cite{IJNT11, BMV12}). We propose an even more general treatment, in which there is no underlying graph, but just an abstract $\ell_1$-minimization problem over an affine subspace~\eqref{l1_min}. Throughout our discussion we  assume that $A$ has rank $m$ (thus in particular~\eqref{l1_min} is feasible). We start by giving the continuous dynamics  and subsequently turn it into a discrete one.

The continuous Physarum dynamics\footnote{More generally, we can define a dynamics solving the above problem with objective replaced by $\sum_{i=1}^n c_i |x_i|$ for any $c\in \R^n_{>0}$. The uniform cost case $c=(1,1,\ldots,1)^\top$ is however the most interesting one (as the non-uniform case reduces to it by scaling). } starts from an arbitrary positive point $w(0)\in \R_{>0}^n,$ its instantaneous velocity vector is given by
\begin{equation}\label{cont_physarum_bp}
\frac{d w(t)}{dt}=\abs{q(t)} - w(t),
\end{equation}
where $q(t)\in \R^n$ is computed as
\begin{equation}
q(t) = W(t) A^\top (AW(t) A^\top)^{-1} b.
\end{equation}
Here $W(t)$ denotes the diagonal matrix $\diag{w(t)}$. In the case of shortest path or the transshipment problem, the vector $q(t)$ corresponds to an electrical flow. It can be equivalently described as the minimizer of weighted $\ell_2$ norm $\inparen{\sum_{i=1}^n \frac{x_i^2}{w_i(t)}}^{1/2}$ over $\{x\in \R^n: Ax=b\}.$
Let us now state an important fact regarding~\eqref{cont_physarum_bp}.
\begin{theorem}\label{thm:existence}
For every initial condition $w(0)\in \R^n_{>0}$ there exists a global solution $w:[0,\infty) \to \R^n_{>0}$ satisfying~\eqref{cont_physarum_bp}.
\end{theorem}
We omit the proof. Let us only mention that the update rule is defined by a locally Lipschitz continuous function, hence the solution to~\eqref{cont_physarum_bp} exists locally. To prove global existence, one needs to show in addition that no solution curve approaches the boundary of $\R^n_{>0}$ in finite time. We refer the reader to~\cite{SVLP15} where a complete proof of existence for a related dynamics is presented. (Though, the case of~\eqref{cont_physarum_bp} is much simpler.)

\subsection{Discrete Physarum dynamics}
We apply Euler's method to discretize the Physarum dynamics from the previous subsection. Pick a small positive step size $h\in (0,1)$ and observe that:
$$w(t+h)-w(t) \approx h \dot{w}(t)  = h(|q(t)|-w(t)).$$
Hence,
$$w(t+h) \approx h|q(t)| +(1-h) w(t).$$
This motivates the following discrete process: pick any $w^{(0)}\in \R_{>0}^n$ and iterate for $k=0,1,\ldots$:
\begin{equation}\label{physarum_discr}
 w^{(k+1)} = h|q^{(k)}| + (1-h) w^{(k)},
 \end{equation}
where as previously $q^{(k)}$ is the result of $\ell_2$-minimization performed with respect to the weights $\inparen{w^{(k)}}^{-1}$. It is given explicitly by the formula  $q^{(k)}=W^{(k)} A^\top (AW^{(k)} A^\top)^{-1}b$.

\section{IRLS vs Physarum}\label{sec:comparison}
In this section we present a proof of  Theorem \ref{thm:main1}. 
When comparing IRLS with the Physarum dynamics one can already see similarities between these two algorithms: both of them are iterative methods which use weighted $\ell_2$-minimization to perform the update. However, apart from this observation no formal connection is apparent. Physarum defines a sequence of strictly positive vectors whose $\ell_1$-norm converges to the optimal $\ell_1$-norm; in particular the iterates are never feasible. The iterates of the IRLS algorithm on the other hand, starting from $k=1$, lie in the feasible region. 

It turns out that the key to understand how these algorithms are  related to each other is by considering them as algorithms working in a larger space: $\R^n \times \R_{> 0}^n$.
  We  show in the subsequent subsections that both algorithms can be seen as maintaining a pair $(y,w)\in \R^n \times \R_{> 0}^n$ such that $y$ satisfies $Ay=b$ and $w$ is the vector of weights guiding the $\ell_2$-minimization. Interestingly, in the original presentation of the Physarum dynamics only the $w$ variable is apparent. In contrast, IRLS keeps track of just the  $y$ variables. This viewpoint allows us to explains how these two algorithms follow essentially the same update rule. 

\subsection{Physarum dynamics and hidden variables}
Recall that Physarum dynamics was defined as starting from some point $w^{(0)}\in \R^n_{>0}$ and evolving according to the rule:
$$w^{(k+1)} = (1-h)w^{(k)} + h|q^{(k)}|$$
with $h\in (0,1)$. Note that $w^{(k)}$ does not quite converge to the optimal solution (it is always positive). The only guarantee we can prove is that $\norm{w^{(k)}}_1$ tends to $\norm{x^\star}_1$ (with $x^\star$ being any optimal solution to~\eqref{l1_min}). Can we recover $x^\star$ from this process? 

Suppose that the starting point $w^{(0)}$ is not arbitrary, but chosen in a specific way. Let $y\in \R^n$ be any solution to $Ay=b$, for instance the least squares solution. For  $w^{(0)}$ we choose any vector $w\in \R_{>0}^n$ which satisfies $|y|\leq w$ entry-wise. Hence, our starting point $w^{(0)}$ belongs to the set:
$$K \defeq \inbraces{w\in \R^n_{> 0}: \exists y\in \R^n \; \; \mathrm{s.t.} \; \;  \inparen{Ay=b \; \;  \mathrm{ and } \; \;  \abs{y}\leq w}}.$$
We now observe a surprising fact.
\begin{fact}
If $\{w^{(k)}\}_{k\in \N}$ is a sequence of points produced by the Physarum dynamics and $w^{(0)}\in K,$ then $w^{(k)}\in K$ for every $k\in \N$.
\end{fact}
\begin{proof}
The proof goes by induction. For $k=0$ the claim holds. Let $k\geq 0$ and consider $w^{(k+1)}$. We have
$$w^{(k+1)} = (1-h)w^{(k)}+h|q^{(k)}|.$$
Hence, if $y$ certifies that $w^{(k)}\in K$ ($Ay=b$ and $\abs{y}\leq w^{(k)}$) then,
$$|(1-h)y_i + h q_i^{(k)}| \leq (1-h)|y_i| + h|q_i^{(k)}| \leq (1-h) w_i^{(k)}+ h |q_i^{(k)}| = w_i^{(k+1)}.$$
In other words, $|(1-h)y+h q^{(k)}|\leq w^{(k+1)}$. This implies that $w^{(k+1)}\in K$ since indeed $$A\inparen{(1-h)y+h q^{(k)}}=b.$$
\end{proof}
\noindent
The above proof actually shows  more. Let $y^{(0)}\in \R^n$ be any point satisfying $Ay^{(0)}=b$ and $w^{(0)}\in \R^{n}_{>0}$ satisfy $\abs{y^{(0)}}\leq w^{(0)}$. If we evolve the pair $(y^{(k)}, w^{(k)})$ according to the rules:
\begin{align*}
w^{(k+1)} &= (1-h)w^{(k)}+h \abs{q^{(k)}},\\
y^{(k+1)} &= (1-h)y^{(k)}+h q^{(k)},
\end{align*}
then $Ay^{(k)}=b$ and $|y^{(k)}|\leq w^{(k)}$ for every $k\in \N$. This implies in particular that 
$$\forall k\in \N \qquad \norm{x^\star}_1 \leq \norm{y^{(k)}}_1 \leq \norm{w^{(k)}}_1.$$
Thus proving convergence of Physarum dynamics is equivalent to showing an appropriate upper bound on $\norm{w^{(k)}}_1$.
The above interpretation of Physarum, as simultaneously evolving two sets of variables is key to understand its connection to IRLS.

\subsection{IRLS as alternate minimization}
We now present IRLS from a (known) alternate minimization viewpoint;  see~\cite{Beck15,Daubechies10}.
Consider the following function $J:\R^{n} \times \R^{n}_{> 0} \to \R$:
$$J(y,w) = \sum_{i=1}^n \frac{y_i^2}{w_i}+\sum_{i=1}^n w_i.$$
$J$ is not well defined when $w_i=0$ for some $i$, but for simplicity let us now ignore this issue.\footnote{The correct way to define $J(y,w)$ in presence of zero entries is the following: whenever $w_i=y_i=0$ we set $\frac{y_i^2}{w_i}=0$ as and whenever $w_i=0$ and $y_i\neq 0$ we define $\frac{y_i^2}{w_i} = +\infty$} It turns out that IRLS can be seen as an alternate minimization method applied to the function $J$. Let us first remark that $J$ is not a convex function. However, when either $y$ or $w$ is fixed, then $J$ is convex as a function of the remaining variables. 

\noindent Consider the following alternate minimization algorithm for $J(y,w)$.
\begin{enumerate}
\item Start with $w^{(0)}=(1,1,\ldots,1)^\top$.
\item For $k=0,1,2,\ldots$:
\begin{itemize}
\item let $y^{(k+1)}$ be the $y$ which minimizes $J(y,w^{(k)})$ over  $y\in \R^n, \; Ay=b$,
\item let $w^{(k+1)}$ be the $w$ which minimizes $J(y^{(k+1)},w)$ over $w\in \R^n_{> 0}$.
\end{itemize}
\end{enumerate}

\noindent The above method tries to minimize the function $J(y,w)$ by alternating between minimization over $y$ with $w$ fixed and minimization over $w$ with $y$ fixed. In general such a scheme is not guaranteed to converge to a global optimum (especially when $J$ is non-convex). 
We now describe what these partial minimization steps correspond to.

\begin{fact}\label{fact_step1}
Suppose that $w\in \R^n_{>0}$ is fixed, then:
$$\argmin_{y} \inbraces{J(y,w): y\in \R^n, Ay=b} =  \argmin_y \inbraces{\sum_{i=1}^n \frac{y_i^2}{w_i}: y\in \R^n, Ay=b}$$
\end{fact}
\noindent
The proof is straightforward; the only point worth noting is  that the second term in $J(y,w)$ does not depend on $y$ and hence does not need to be taken into account. We now analyze the second step.

\begin{fact}\label{fact_step2}
Suppose that $y\in \R^n$ is fixed and $y_i\neq 0$ for all $i\in [n]$ then:
$$\argmin \inbraces{J(y,w): w\in \R^n_{> 0}} =  \abs{y}.$$
\end{fact}
\noindent
In the above we make a simplifying assumption that no entry of $y$ is zero. This is not crucial, but to drop this assumption, a more rigorous treatment is necessary. It can be done, at a cost of making the notation less transparent.

\begin{proof}
We would like to minimize
$$J(y,w)= \sum_{i=1}^n \frac{y_i^2}{w_i}+\sum_{i=1}^n w_i$$
for a fixed $y$. Note that the above function is separable, hence it suffices to minimize
$$\frac{y_i^2}{w_i} + w_i$$
separately for every $i$. By a simple calculation one can find that the above expression is minimized when
$$w_i = |y_i|.$$
\end{proof}
\noindent
Note now that Facts~\ref{fact_step1} and \ref{fact_step2} together imply that the sequence $y^{(1)}, y^{(2)}, \ldots$ resulting from alternate minimization is the same as that produced by IRLS.
As a byproduct, we also obtain that $\norm{y^{(k)}}_1$ is non-increasing with $k$, because:
\begin{align*}
J\inparen{y^{(k)},w^{(k)}} &= \sum_{i=1}^n \frac{\abs{y_i^{(k)}}^2}{w_i^{(k)}}+\sum_{i=1}^n w_i^{(k)}\\
& = \sum_{i=1}^n \frac{\abs{y_i^{(k)}}^2}{\abs{y_i^{(k)}}} + \sum_{i=1}^n  \abs{y_i^{(k)}} \\
&= 2\norm{y^{(k)}}_1.
\end{align*}
Of course $J\inparen{y^{(k)},w^{(k)}}$ is non-increasing for $k\geq 1$, hence $\norm{y^{(k)}}_1$ is non-increasing as well.
\subsection{Comparing IRLS with Physarum}
In this subsection we conclude our previous considerations by giving a unifying viewpoint on Physarum and IRLS. In fact both of them can be seen as algorithms working in the $2n$-dimensional space $\Gamma=\R^n \times \R^n_{> 0}$. Let us state both algorithms in a similar form.

\noindent\begin{minipage}[t]{8.08cm}
\vspace{0pt}  
\begin{algorithm}[H]

 \KwData{$A\in \R^{m\times n}, b\in \R^m$}
 $w^{(0)}=(1,1,...,1)^\top \in \R^n$\;
 \For{$k=0,1,2,...$}{
  $q=\argmin \sum_{i=1}^n \frac{x_i^2}{w_i^{(k)}} \quad \mathrm{ s.t. } \; \; Ax=b$\;
  $y^{(k+1)} = q$\;
  $w^{(k+1)}=|q|$\;
 }
 \caption{IRLS} 
\end{algorithm}
\end{minipage}%
\hspace{0.3cm}
\begin{minipage}[t]{8.08cm}
\vspace{0pt}  
\begin{algorithm}[H]
 \KwData{$A\in \R^{m\times n}, b\in \R^m$}
 $w^{(0)}=(1,1,...,1)^\top \in \R^n$, $h\in (0,1)$\;
 \For{$k=0,1,2,...$}{
  $q=\argmin \sum_{i=1}^n \frac{x_i^2}{w_i^{(k)}} \quad \mathrm{ s.t. }\; \; Ax=b$\;
  $y^{(k+1)} = (1-h) y^{(k)} +hq$\;
  $w^{(k+1)}=(1-h) w^{(k)} +h|q|$\;
 }
 \caption{Physarum} 
\end{algorithm}
\end{minipage}
\vspace{2mm}

\noindent
The above comparison yields a clear connection between IRLS and Physarum. Let us define a vector field $F:\Gamma \to \R^n \times \R^n$ by the following formula:
\begin{align}
\begin{split}
F(y,w) &\defeq (q-y,|q|-w),\\
 q& \defeq \argmin_{x\in \R^n} \sum_{i=1}^n \frac{x_i^2}{w_i} \qquad \mathrm{s.t. } \; \;  Ax=b.
\end{split}
\end{align}
The IRLS algorithm given a point $(x,w)\in \R^n \times \R^n_{> 0}$ simply moves along the vector $F(x,w)$ to the new point $(x,w)+F(x,w)$, while Physarum moves to a point on the interval between $(x,w)$ and $(x,w)+F(x,w)$. For this reason Physarum can be seen as a damped variant of IRLS. 

Let us now define two interesting subsets of $\Gamma$ (see Figure~\ref{fig:gamma} for a one-dimensional example)
\begin{align*}
P&\defeq \inbraces{(y,w)\in \Gamma: Ay=b, \abs{y}\leq w},\\
\bar{P} & \defeq \inbraces{(y,w)\in \Gamma: Ay=b, \abs{y}=w}.
\end{align*}
\begin{figure}[!ht]
  
  \centering
    \includegraphics[width=0.7\textwidth]{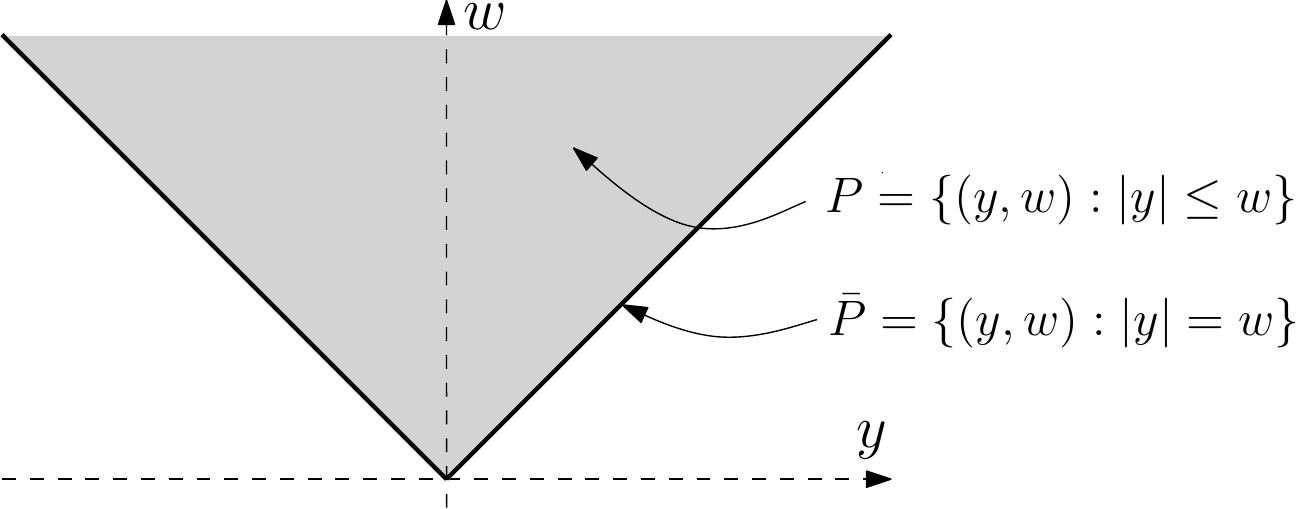}
    \caption{Illustration of $\Gamma, P$ and $\bar{P}$ for $n=1$.}
    \label{fig:gamma}
\end{figure}
IRLS can be seen as a discrete dynamical system defined over $\bar{P}$, while Physarum initialized at a point $w^{(0)}\in P$ stays in $P$, for any choice of $h\in (0,1)$.\footnote{Physarum initialized at a point outside of $P$ converges to $P$.} Interestingly $\bar{P}$ is a non-convex set, which is the boundary of $P$ (in contrast $P$ is convex).

In the next section we prove that Physarum never faces the issue of $w^{(k)}_i$ being zero for some $i$, indeed $w^{(k)}>0$ for every $k$, which follows from the fact that $h<1$. In contrast, Physarum with $h=1$ is equivalent to IRLS, where this happens frequently.

\section{Convergence and Complexity of Physarum Dynamics}\label{sec:convergence}
In this section we study convergence of Physarum dynamics. The analysis is based on ideas developed in~\cite{BBDKM13,SV15,SVLP15}. Specifically, we prove the following theorem, whose informal version appeared as Theorem \ref{thm:main2}. Let  $\alpha \defeq \max \inbraces{|\det(A')|:A' \; \mathrm{is \;  a \;  square  \; submatrix  \; of \; } \; A}$.

\begin{theorem}\label{thm:convergence}
Suppose $w^{(0)}$ was chosen to satisfy $\abs{y^{(0)}}\leq w^{(0)}$ for some $y^{(0)}\in \R^n$ such that $Ay^{(0)}=b$. Furthermore assume $w_i^{(0)}\geq 1$ for every $i\in [n]$ and $\norm{w^{(0)}}_1 \leq M \norm{x^\star}_1$ for some $M\in \R$. Let $\eps\in (0,\nfrac{1}{2})$ and $h\leq \frac{\eps}{40n^2\alpha^2}$. Then after $k=O\inparen{ \frac{\ln M + \ln \norm{x^\star}_1}{h\eps^2}}$ steps $\norm{w^{(k)}}_1 \leq (1+\eps) \norm{x^\star}_1$ and one can easily recover a vector $y^{(k)}$ such that $Ay^{(k)}=b$ and $\norm{y^{(k)}}_1 \leq \norm{w^{(k)}}_1$.
\end{theorem}

\noindent
Few comments are in order. The assumptions about the starting point $w^{(0)}$, we made in the statement, are not necessary for convergence. However, they greatly simplify the proofs and make it easy to recover a close to optimal feasible solution to~\eqref{l1_min_intro}. The choice of the step size $h$  follows directly from our analysis and is not likely to be optimal. Experiments suggest that the claimed iteration bound should hold even for $h$ being a small constant (not depending on the data).

\medskip
\parag{Assumptions, notation and simple facts.}
Motivated by the observation about hidden variables made in Section~\ref{sec:comparison}, we  assume that the starting point $w^{(0)}$ is chosen in such a way that $w^{(0)}>0$ and $\abs{y^{(0)}}\leq w^{(0)}$ for some $y^{(0)}\in \R^n$ such that $Ay^{(0)}=0$. Recall that in that case, for every $k$ we are guaranteed existence of a feasible $y^{(k)}$ with $\abs{y^{(k)}}\leq w^{(k)}$. Moreover, these $y^{(k)}$ are easy to find. One particular choice of $y^{(0)}$ and $w^{(0)}$ could be the least squares solution to $Ax=b$ and $w_i^{(0)}=\abs{y_i^{(0)}}+1$ respectively.

Let us now verify  that $w^{(k)}>0$ at all steps and hence that the Physarum dynamics is well defined.
\begin{lemma}
For every $k$, $w^{(k)}\in \R_{>0}^n$.
\end{lemma}
\begin{proof}
The proof goes via simple induction. For $k=0$ the claim is valid by assumption that $w^{(0)}>0$, next for $k\geq 0$ we have:
$$w_i^{(k+1)} = (1-h)w_i^{(k)}+h\abs{q_i^{(k)}}>0$$
because $h\in (0,1)$.
\end{proof}

\noindent
The above lemma shows in particular that the weighted $\ell_2$-minimization problem solved in every step indeed has  a unique optimal solution.

For the convergence proof let us fix $x^\star\in \R^n$ to be any optimal solution to our $\ell_1$-minimization problem~\eqref{l1_min}.  Without loss of generality we may assume that $x^\star\geq 0$ (if not, multiply by $(-1)$ all the columns of $A$ which correspond to negative entries in $x^\star$, it does not change the problem neither the sequence produced by Physarum).
To track the convergence process of Physarum the two following quantities are useful:
\begin{enumerate}
\item $E(k) = \sum_{i=1}^n \frac{\inparen{q_i^{(k)}}^2}{w_i^{(k)}},$
\item $\barr(k) = \sum_{i=1}^n x_i^\star \ln w_i^{(k)}.$
\end{enumerate}

\medskip
\parag{A technical lemma.}
The following technical lemma from \cite{SVLP15} is particularly useful in our setting. We state the lemma together with a proof to make the paper self-contained. For a version with quantitative bounds we refer the reader to~\cite{SVLP15}. 
\begin{lemma}\label{lemma:technical}
Consider a weight vector $w\in \R_{>0}^n$, then the matrix $L=AWA^\top$ is invertible and:
$$\forall i,j\in [n] \quad |a_i^\top L^{-1} a_j| \leq \frac{\alpha}{w_i}$$
where $\alpha \in \R$ is a constant which depends solely on $A$.
\end{lemma}
\begin{proof}
Take any weight vector $w\in \R^n_{>0}$ and pick $i,j\in [n]$. By symmetry it is enough to establish the above bound with $w_i$ replaced by $w_j$. We first note that:
$$w_j a_j a_j^\top \preceq  \sum_{k=1}^n w_k a_k a_k^\top =L,$$
where $\preceq$ is the Loewner ordering. By testing the above on the vector $L^{-1}a_j$, we obtain:
$$(L^{-1}a_j)^\top w_j a_j a_j^\top (L^{-1}a_j) \leq (L^{-1}a_j)^\top L (L^{-1}a_j).$$
By a simple calculation the above yields: $$a_j^\top L^{-1}a_j\leq \frac{1}{w_j}.$$ In the remaining part of the argument we  show that $|a_i^\top L^{-1}a_j| \leq \alpha a_j^\top L^{-1}a_j$, where $\alpha$ will be specified later.

  If $a_i^\top L^{-1}a_j=0$ then there is nothing to prove. Otherwise, let us call $p=L^{-1}a_j$, we may assume without loss of generality that $a_k^\top p\geq 0$ for every $k\in [n]$ (we may reduce our problem to this case by multiplying some $a_k$'s by $(-1)$). Because $Lp=a_j$, we obtain:
$$a_j = Lp = \inparen{\sum_{i=k}^n w_k a_ka_k^\top}p = \sum_{i=k}^n w_k (a_k^\top p) a_k$$
Hence we have just obtained a representation of $a_j$ as a conic combination of $a_1, a_2, \ldots, a_n$. Moreover, the set $$S_j^i=\{s\in \R^n: s\geq 0, \; \sum_{k=1}^n s_k a_k=a_j, \; s_i>0\}$$ is non-empty. Let us take an element $r\in \R^n$ of $S_j^i$ which maximizes $r_i$. Note that $S_j^i$ depends solely on $A$. Hence there is some lower bound for $r_i$, let us call it $\frac{1}{\alpha}$ for some $\alpha>0$. We obtain:
$$a_j^\top L^{-1} a_j = p^\top a_j = p^\top \inparen{\sum_{k=1}^n r_k a_k} = \sum_{k=1}^n r_k (p^\top a_k)\geq r_i p^\top a_i\geq \frac{1}{\alpha} a_iL^{-1}a_j.$$
\end{proof}

\begin{remark}\label{remark:alpha}
From now on we  state all bounds with respect to $\alpha$ obtained in the above lemma.  \cite{SVLP15} shows that if $A$ is a matrix with integer entries then $\alpha$ can be chosen to be:
$$ \max \inbraces{|\det(A')|:A' \; \mathrm{is \;  a \;  square  \; submatrix  \; of \; } \; A}.$$
In general, one can bound $\alpha$ in terms of the maximum absolute value of all entries of $(A')^{-1}$ over all invertible square submatrices $A'$ of $A$.
\end{remark}

\noindent
The following corollary  is used multiple times in the convergence proof. Recall that we work under the assumption that $w^{(0)}\geq \abs{y^{(0)}}$ for some $y^{(0)}\in \R^n$ with $Ay^{(0)}=b$.

\begin{corollary}\label{cor:bounded_potentials}
Suppose that $\inbraces{w^{(k)}}_{k\in \N}$ is the sequence produced by Physarum and $\inbraces{q^{(k)}}_{k\in \N}$ is the corresponding sequence of weighted $\ell_2$-minimizers. Then, for every $k$:
$$\forall i\in [n] \quad \frac{\abs{q_i^{(k)}}}{w_i^{(k)}} \leq n\alpha,$$
for $\alpha$ being the same constant as in Lemma~\ref{lemma:technical}.
\end{corollary}
\begin{proof}
Let $L=AW^{(k)}A^\top$ (note that both $L$ and $L^{-1}$ are symmetric matrices), then:
$$q^{(k)} = W^{(k)} A^\top L^{-1}b.$$
Hence:
$$\frac{|q_i^{(k)}|}{w_i^{(k)}} = \abs{a_i^\top L^{-1}b}.$$
Recall that $b=Ay^{(k)}$ where $\abs{y^{(k)}}\leq w^{(k)}$, hence:
\begin{align*}
\frac{|q_i^{(k)}|}{w_i^{(k)}} &= \abs{a_i^\top L^{-1}A y^{(k)}}\\
 & \leq \sum_{j=1}^n  \abs{y_j^{(k)} \cdot a_i^\top L^{-1}a_j}\\
 & = \sum_{j=1}^n  \abs{y_j^{(k)}}\cdot \abs{ a_j^\top L^{-1}a_i}\\
 & \stackrel{\mbox{Lemma \ref{lemma:technical}}}{\leq}  \sum_{j=1}^n  \abs{y_j^{(k)}}\cdot \frac{\alpha}{w_i^{(k)}}\\
 &\leq n\alpha.
\end{align*}
\end{proof}

\medskip
\parag{Analysis of potentials.}
\begin{lemma}\label{lemma:norm_drop}
For every $k\in \N$ we have $\norm{w^{(k+1)}}_1 \leq \norm{w^{(k)}}_1$. Furthermore, if for some $\eps \in \inparen{0,\frac{1}{2}}$ we have $\norm{w^{(k)}}>\inparen{1+\frac{\eps}{3}}E(k)$ then $\norm{w^{(k+1)}}\leq  \inparen{1-\frac{h\eps}{8}} \norm{w^{(k)}}$.
\end{lemma}

\begin{proof}
We have:
$$ \norm{w^{(k)}}_1 - \norm{w^{(k+1)}}_1 = h \sum_{i=1}^n \inparen{w_i^{(k)}-\abs{q_i^{(k)}}}=h\inparen{ \norm{w^{(k)}}_1-\norm{q^{(k)}}_1}.$$
Furthermore:
$$\norm{q^{(k)}}_1 = \sum_{i=1}^n \abs{q_i^{(k)}} = \sum_{i=1}^n \sqrt{w_i^{(k)}} \frac{|q_i^{(k)}|}{\sqrt{w_i^{(k)}}},$$
\noindent by applying the Cauchy-Schwarz inequality, we obtain:
$$
 \sum_{i=1}^n \sqrt{w_i^{(k)}} \frac{|q_i^{(k)}|}{\sqrt{w_i^{(k)}}}\leq  \norm{w^{(k)}}_1^{1/2} \cdot E(k)^{1/2}.$$

\noindent Thus we finally get:
$$h\inparen{\norm{w^{(k)}}_1-\norm{q^{(k)}}_1}\geq h \norm{w^{(k)}}_1^{1/2} \inparen{\norm{w^{(k)}}_1^{1/2} - E(k)^{1/2}}.$$ 
Since $q^{(k)}$ minimizes the weighted $\ell_2$ norm over the subspace $Ax=b$, we obtain:
$$E(k) = \sum_{i=1}^n \frac{\inparen{q_i^{(k)}}^2}{w_i^{(k)}} \leq \sum_{i=1}^n \frac{\inparen{y_i^{(k)}}^2}{w_i^{(k)}}\leq \sum_{i=1}^n \frac{\inparen{w_i^{^{(k)}}}^2}{w_i^{(k)}} = \norm{w^{(k)}}_1.$$
Hence the first part of the lemma is proved. Assume now $\norm{w^{(k)}}>\inparen{1+\frac{\eps}{3}}E(k)$. We get:
\begin{align*}
\norm{w^{(k)}}_1 - \norm{w^{(k+1)}}_1 &\geq h  \norm{w^{(k)}}_1^{1/2} \inparen{\norm{w^{(k)}}_1^{1/2} - E(k)^{1/2}}\\
&\geq h \inparen{1-\inparen{1+\frac{\eps}{3}}^{-1/2}} \norm{w^{(k)}}_1 .
\end{align*}
It remains to note that $1-\inparen{1+\frac{\eps}{3}}^{-1/2}\geq \frac{\eps}{8}$.
\end{proof}

To analyze the behavior of $\barr(k)$ we  use the following elementary inequality:
\begin{align}\label{ineq:log}
x-x^2\leq \ln(1+x)& \leq x 
\end{align}
which is valid for all $-\frac{1}{2}\leq x\leq \frac12$. Let us also state the following useful fact.
\begin{fact}\label{fact:energy}
Let $w\in \R^{n}_{>0}$ and $q\in \R^n$ be the solution to $\displaystyle \min_{x\in \R^n} \inbraces{\sum_{i=1}^n \frac{x_i^2}{w_i}: Ax=b}$. Then:
$$b^\top L^{-1} b = \sum_{i=1}^n \frac{q_i^2}{w_i},$$
where $L=AWA^\top$.
\end{fact}
\begin{proof}
We use the explicit formula $q=WA^\top L^{-1}b$. Note that $ \sum_{i=1}^n \frac{q_i^2}{w_i} = q^\top W^{-1} q$ and hence:

\begin{align*}
q^\top W^{-1} q &= b^\top L^{-1}AW W^{-1}WA^\top L^{-1}b \\
&= b^\top L^{-1}(AWA^\top) L^{-1}b\\
& =b^\top L^{-1}b .
\end{align*}
\end{proof}
\noindent
We continue with a lemma describing the behavior of $\barr(k)$.

\begin{lemma}\label{lemma:barr}
Suppose that $h\leq \frac{\eps}{40\cdot (n\alpha)^2}$, then for every $k$ it holds that $$\barr(k+1)\geq \barr(k) +h\inparen{E(k) -\inparen{1+\frac{\eps}{10}}\norm{x^\star}_1}.$$
\end{lemma}
\begin{proof}
We have:
\begin{align*}
\barr(k+1) - \barr(k)& = \sum_{i=1}^n x^\star_i \ln \frac{w_i^{(k+1)}}{w_i^{(k)}}\\
&=\sum_{i=1}^n x^\star_i \ln \inparen{1+h\inparen{\frac{\abs{q_i^{(k)}}}{w_i^{(k)}}-1}}
\end{align*}
We apply the left-hand side of~\eqref{ineq:log} to every summand. This is possible by our assumption $x^\star\geq 0$. For simplicity let $z_i \defeq \inparen{\frac{\abs{q_i^{(k)}}}{w_i^{(k)}}-1}$. We obtain:
\begin{align}
\begin{split}\label{ineq:barr}
\barr(k+1) - \barr(k)& \geq  \sum_{i=1}^n x^\star_i (hz_i - h^2z_i^2)\\
&=h\sum_{i=1}^n x^\star_i z_i - h^2\sum_{i=1}^n x^\star_iz_i^2
\end{split}
\end{align}

We analyze the linear term and quadratic term separately. We have:
$$\sum_{i=1}^n x^\star_i z_i = \sum_{i=1}^n x^\star_i \inparen{\frac{\abs{q_i^{(k)}}}{w_i^{(k)}}-1} = \sum_{i=1}^n x^\star_i \inparen{\frac{\abs{q_i^{(k)}}}{w_i^{(k)}}} - \norm{x^\star}_1.$$
We lower-bound the first order term:
\begin{align*}
\sum_{i=1}^n x^\star_i \inparen{\frac{\abs{q_i^{(k)}}}{w_i^{(k)}}} &\geq \sum_{i=1}^n x^\star_i \frac{q_i^{(k)}}{w_i^{(k)}} \\ &=(x^\star)^\top \inparen{W^{(k)}}^{-1} q^{(k)} \\
&= (x^\star)^\top \inparen{W^{(k)}}^{-1} W^{(k)} A^\top L^{-1} b\\
&=(x^\star)^\top  A^\top L^{-1} b \\
&= b^\top L^{-1}b
\end{align*}
where $L=AW^{(k)}A^\top$. The above, together with Fact~\ref{fact:energy} give:
$$\sum_{i=1}^n x^\star_i \inparen{\frac{\abs{q_i^{(k)}}}{w_i^{(k)}}}\geq b^\top L^{-1}b = E(k).$$
Thus we have obtained:
\begin{equation}\label{ineq:first_order}
\sum_{i=1}^n x^\star_i z_i \geq E(k) - \norm{x^\star}_1.
\end{equation}
To bound the quadratic term in~\eqref{ineq:barr} we just apply Corollary~\ref{cor:bounded_potentials}:
$$\sum_{i=1}^n x^\star_iz_i^2 \leq \sum_{i=1}^n x^\star_i\inparen{n\alpha-1}^2\leq (2n\alpha)^2 \norm{x^\star}_1.$$
We combine~\eqref{ineq:barr} with our bounds on first and second order terms to obtain:

\begin{align*}
\barr(k+1) - \barr(k) &\geq h(E(k) - \norm{x^\star}_1)-h^2(2n\alpha)^2 \norm{x^\star}_1\\
&\geq h(E(k) - \norm{x^\star}_1)-h\cdot \frac{\eps}{10}\cdot \norm{x^\star}_1.
\end{align*}
\end{proof}

\medskip
\parag{Convergence proof.}
We are ready to prove the main result.
\begin{proofof}{of Theorem~\ref{thm:convergence}}
We would like to count the number of steps till the first moment when $\norm{w^{(k)}}_1 \leq (1+\eps)\norm{x^\star}_1$. From Lemma~\ref{lemma:norm_drop} the $\ell_1-$norm of $w^{(k)}$ is non-increasing with $k$ and whenever $\norm{w^{(k)}}_1>\inparen{1+\frac{\eps}{3}}E(k)$, $\norm{w^{(k)}}_1$ decreases by a multiplicative factor of $(1-\frac{h\eps}{8})$. This means that there can be at most 
$$\log_{(1-h\eps)^{-1}}\inparen{\frac{M}{1+\eps}} = O\inparen{\frac{\ln M}{h\eps}}$$
such steps. What about steps for which $\norm{w^{(k)}}_1\leq \inparen{1+\frac{\eps}{3}}E(k)$? We obtain:
$$(1+\eps)\norm{x^\star}_1 \leq \norm{w^{(k)}}_1 \leq  \inparen{1+\frac{\eps}{3}}E(k).$$
This in particular  implies that:
$$E(k) \geq\inparen{1+\frac{\eps}{2}}\norm{x^\star}_1.$$
We apply Lemma~\ref{lemma:barr} to conclude that in such a case:
$$\barr(k+1) \geq \barr(k) + \frac{h\eps}{3} \norm{x^\star}_1.$$
Let us now analyze how $\barr(k)$ can change throughout steps. We start with $\barr(0)\geq 0$ (since $w_i^{(0)}\geq 1$ for every $i\in [n]$) and $\barr(k)$ is upper bounded by $\norm{x^\star}_1\cdot  (\ln M+ \ln \norm{x^\star}_1)$ (this holds because $\norm{w^{(k)}}_1 \leq \norm{w^{(0)}}_1 \leq M \norm{x^\star}_1$). At every step when $\norm{w^{(k)}}_1>\inparen{1+\frac{\eps}{3}}E(k)$ the largest possible drop of $\barr(k)$ is (by Lemma~\ref{lemma:barr}) upper-bounded by:
$$h\inparen{1+\frac{\eps}{10}}\norm{x^\star}_1 \leq 2h\norm{x^\star}_1.$$
Note that by the reasoning above there are at most $O\inparen{\frac{\ln M}{h\eps}}$ such steps. On the other hand, if $\norm{w^{(k)}}_1\leq \inparen{1+\frac{\eps}{3}}E(k)$ then $\barr(k)$ increases by at least: $ \frac{h\eps}{3} \norm{x^\star}_1.$ This means that the total drop of $\barr(k)$ over the whole computation is at most:
$$O\inparen{\frac{\ln M}{\eps}\norm{x^\star}_1}.$$
Hence the number of steps in which $\norm{w^{(k)}}_1\leq \inparen{1+\frac{\eps}{3}}E(k)$ is at most:
$$O\inparen{\frac{\frac{\ln M}{\eps}\norm{x^\star}_1 + \norm{x^\star}_1\cdot  \inparen{\ln M+ \ln \norm{x^\star}_1}}{\frac{h\eps}{3} \norm{x^\star}_1}} = O\inparen{ \frac{\ln M + \ln \norm{x^\star}_1}{h\eps^2}}.$$
\end{proofof}

\appendix

\section{Example for Non-convergence of IRLS}\label{app:nonconvergence}
We present an example instance for which IRLS fails to converge to the optimal solution. More precisely we prove the following.
\begin{theorem}\label{thm:nonconvergence}
There exists an instance $(A,b)$ of the basis pursuit problem~\eqref{l1_min_intro} and a feasible, strictly positive point $y\in \R^n_{>0}$ such that if IRLS is initialized at $y^{(0)}= y$ (and $\{y^{(k)}\}_{k\in \N}$ is the sequence produced by IRLS) then $\norm{y^{(k)}}_1 $ does not converge to the optimal value.
\end{theorem}

\noindent 
   The proof is based on the simple observation that if IRLS reaches a point $y^{(k)}$ with $y^{(k)}_i=0$ for some $k\in \N$, $i\in [n]$ then $y^{(l)}_i=0$ for all $l>k$. 

Let us consider an undirected graph $G=(V,E)$ with $V=\{u_0, u_1, ...,u_6,u_7\}$ and let $s=u_0$, $t=u_7$. $G$ is depicted in Figure~\ref{fig:graph}.

\begin{figure}[!ht]
  
  \centering
    \includegraphics[width=0.6\textwidth]{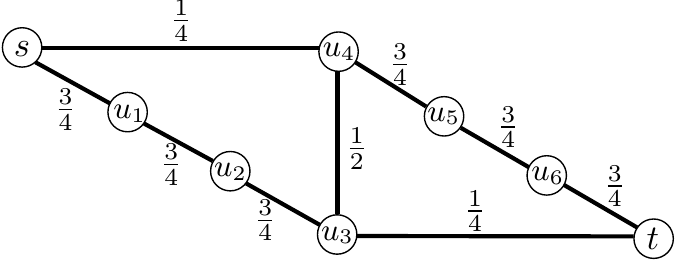}
    \caption{The graph $G$ together with a feasible solution $y\in \R^V$.}\label{fig:graph}
\end{figure}

\noindent
We define $A\in \R^{V\times E}$ to be the signed incidence matrix of $G$ with edges directed according to increasing indices, let $b\defeq e_t - e_s=(-1,0,0,0,0,0,0,1)^\top$. Then the following problem:
$$\min \; \norm{x}_1 \qquad \mathrm{s.t. }\; \; Ax=b$$
is equivalent to the shortest $s-t$ path problem in $G$. The unique optimal solution is the path $s-u_4-u_3-t$. In particular, the edge $(u_3,u_4)$ is in the support of the optimal vector.

\begin{claim}\label{claim:graph}
Let $y\in \R^E$ be a feasible point given in the Figure~\ref{fig:graph}, i.e. $y_{u_0u_1}=y_{u_1u_2}=y_{u_2u_3}=y_{u_4u_5}=y_{u_5u_6}=y_{u_6u_7}=\frac{3}{4}$, $y_{u_0u_4}=y_{u_3u_7}=\frac14$ and $y_{u_3u_4}=\frac{1}{2}$. IRLS initialized at $y$ produces in one step a point $y'$ with $y'_{u_3y_4}=0$. 
\end{claim}

\noindent
The above claim implies that IRLS initialized at $y$ (which has full support) does not converge to the optimal solution, which has $1$ in the coordinate corresponding to $u_3u_4$. Thus to prove Theorem~\ref{thm:nonconvergence} it suffices to show Claim~\ref{claim:graph}.
\begin{proofof}{of Claim~\ref{claim:graph}}
IRLS chooses the next point $y'\in \R^E$ according to the rule:
$$y' = \argmin_{y\in \R^E} \sum_{e\in E} \frac{x_e^2}{y_e} \qquad \mathrm{s.t. } \; \; Ax=b$$
which is the same as the unit electrical $s-t$ flow in $G$ corresponding to edge resistances $\frac{1}{y_e}$.\footnote{This is due to the fact that electrical flows minimize energy.} One can easily see that in such electrical flow the potentials of $u_4$ and $u_3$ are equal (the paths $s-u_4$ and $s-u_1-u_2-u_3$ have equal resistances), hence the flow through $(u_3,u_4)$ is zero. 
\end{proofof}

\bibliographystyle{alpha}
\bibliography{references}
\end{document}